%% file: arXiv.tex
\documentclass[11pt, letter]{article}  
\usepackage[margin=1in]{geometry}
\usepackage{graphicx}
\usepackage{mathptmx}     
\usepackage{amsmath,amssymb,xcolor, enumerate,amsthm,mathtools}
\usepackage{bm}
\usepackage[utf8]{inputenc}
\usepackage{hyperref}
\usepackage[linesnumbered,algoruled,boxed,lined]{algorithm2e}
\usepackage{subcaption}
\usepackage{standalone}
\usepackage{authblk}
\usepackage{tikz}
\usetikzlibrary{intersections}
\usetikzlibrary{arrows, calc, shapes}
\usepackage{enumitem}
\usepackage{hyperref}
\usepackage[linecolor=red,backgroundcolor=red!25,bordercolor=red]{todonotes}
\usepackage[backend=bibtex]{biblatex}
\usepackage{nicefrac}
\usepackage{thmtools}
\usepackage{thm-restate}

\bibliography{bib}

\mathchardef\mhyphen="2D
\newcommand{\R}{\mathbb{R}}
\newcommand{\Z}{\mathbb{Z}}

\newcommand{\cP}{\mathcal{P}}

\newcommand{\cT}{\mathcal{T}}
\newcommand{\cR}{\mathcal{R}}
\newcommand{\cL}{\mathcal{L}}
\newcommand{\fmin}{f_{\min}}

\newcommand{\dynN}{\mathcal{N}}

\newcommand{\N}{\mathcal{N}}

\newcommand{\MCF}{\operatorname{MCF}}

\newcommand{\mincut}{\operatorname{cut}}
\newcommand{\cut}{\operatorname{cut}}

\newcommand{\terminals}{S^+ \cup S^-}

\makeatletter
\DeclareRobustCommand{\cev}[1]{%
  {\mathpalette\do@cev{#1}}%
}
\newcommand{\do@cev}[2]{%
  \vbox{\offinterlineskip
    \sbox\z@{$\m@th#1 x$}%
    \ialign{##\cr
      \hidewidth\reflectbox{$\m@th#1\vec{}\mkern4mu$}\hidewidth\cr
      \noalign{\kern-\ht\z@}
      $\m@th#1#2$\cr
    }%
  }%
}
\makeatother

\newtheorem{theorem}{Theorem}

\newtheorem{lemma}[theorem]{Lemma}

\newtheorem{obs}{Observation}
\theoremstyle{definition}

\title{A Faster Algorithm for Quickest Transshipments 
via an Extended Discrete Newton Method}
\date{}

\author[1]{Miriam Schlöter}
\author[2]{Martin Skutella}
\author[2]{Khai Van Tran}

\affil[1]{Department of Mathematics, ETH Zürich, Switzerland, \texttt{miriam.schloeter@ifor.math.ethz.ch}}
\affil[2]{Institute of Mathematics, TU Berlin, Germany,
\texttt{\{skutella,kvtran\}@math.tu-berlin.de}}

\begin{document}
\maketitle
\begin{abstract}
The Quickest Transshipment Problem is to route flow as quickly as possible from sources with supplies to
sinks with demands in a network with capacities and transit times on the arcs. It is of fundamental
importance for numerous applications in areas such as logistics, production, traffic, evacuation, and
finance. More than 25 years ago, Hoppe and Tardos presented the first (strongly) polynomial-time algorithm for this
problem. Their approach, as well as subsequently derived algorithms with strongly polynomial running
time, are hardly practical as they rely on parametric submodular function minimization via Megiddo's
method of parametric search. The main contribution of this paper is a considerably faster algorithm for
the Quickest Transshipment Problem that instead employs a subtle extension of the Discrete Newton Method.
This improves the previously best known running time of $\tilde{O}(m^4k^{14})$ to $\tilde O(m^2k^5+m^3k^3+m^3n)$, where $n$ is the number of nodes, $m$ the number of arcs, and $k$ the number of sources and sinks.
\end{abstract}

\thispagestyle{empty}
\pagebreak
\setcounter{page}{1}

\section{Introduction}
	\input{0_Introduction}

\section{Preliminaries.}\label{secIntro}
	\input{1_Preliminaries}

\section{A Simple Algorithm Akin to Discrete Newton}\label{sec:evac}
	\input{2_Evacuation_Problem}

\section{A Refined Algorithm with Long Jumps}\label{sec:main}
	\input{3_Main_Alg}


\paragraph{Acknowledgement.} The authors are much indebted to Tom McCormick for very interesting discussions and, in particular, for diligently emphasizing the promise of discrete Newton.
\pagebreak
\printbibliography
\end{document}

%% file: 0_Introduction.tex

Time is a critical resource in many network routing problems arising in road, pedestrian, rail, or air
traffic control, including evacuation planning as one important example~\cite{HamacherTjandra02a}. Network
flows over time capture the essence of these applications as they model the variation of flow along arcs
over time as well as the delay experienced by flow traveling at a given pace through the network. Other
application areas include, for instance, production systems, communication networks, and financial
flows~\cite{Aronson89,PowellJO95}.

\paragraph{Maximum Flows Over Time.}
The study of flows over time goes back to the work of Ford and Fulkerson~\cite{Ford1958,Ford1962}. Given
a directed graph~$D=(V,A)$ with integral \emph{arc capacities}~$u\in\Z_{\geq 0}^A$ and integral
\emph{arc transit times}~$\tau\in\Z_{\geq 0}^A$, they consider the \emph{Maximum Flow Over Time
Problem}, that is, to send the maximum possible amount of flow from a source node~$s$ to a sink node~$t$
within a given time horizon~$\theta\in\Z_{>0}$. Here, the capacity bounds the \emph{rate} at which flow
may enter an arc, while its transit time specifies how much time it takes for flow to travel from the
tail of an arc to its head.

Ford and Fulkerson observe that this problem can be reduced to a static maximum flow problem in an
exponentially large \emph{time-expanded network}, whose node set consists of~$\theta$ copies of the given
node set~$V$. More importantly, Ford and Fulkerson show that a path decomposition of a static min-cost
$s$-$t$-flow (with arc transit times as costs) in the given network yields a maximum flow over time by
repeatedly sending flow along these $s$-$t$-paths at the corresponding flow rates. In particular, using
the min-cost flow algorithm of Orlin~\cite{Orlin1993} in conjunction with Thorup's shortest path
algorithm~\cite{Thorup04}, the maximum flow value can be computed in~$O(m\log n (m + n \log\log n))$
time, where~$n$ is the number of nodes and~$m$ the number of arcs of the network. In the following we
use~$O(\MCF)$ to refer to this min-cost flow running time.

\paragraph{Quickest Flows.} 
Closely related to the Maximum $s$-$t$-Flow Over Time Problem is the \emph{Quickest Flow Problem}:
Instead of fixing the time horizon, we fix the amount of flow to be sent and ask for an $s$-$t$-flow over
time with minimal time horizon~$\theta^*$. Burkard, Dlaska, and Klinz~\cite{Burkard1993} observe that the
Quickest $s$-$t$-Flow Problem can be solved in strongly polynomial time by incorporating the algorithm of
Ford and Fulkerson in Megiddo’s parametric search framework~\cite{Megiddo1979}. Lin and
Jaillet~\cite{Lin2015} present a cost-scaling algorithm that solves the problem in
$O(nm\log(n^2/m)\log(nC))$ time, where~$C$ is the maximum arc transit time, thus matching the min-cost
flow running time of Goldberg and Tarjan's cost-scaling algorithm~\cite{Goldber1990}. Refining Lin and
Jaillet's approach, Saho and Shigeno~\cite{Saho2017} achieve an algorithm with the currently fastest known
strongly polynomial running time $O(m^2n\log^2 n)$.

\paragraph{Quickest Transshipments.}
The \emph{Quickest Transshipment Problem} generalizes the Quickest Flow Problem to the setting with
multiple source and sink nodes with given supplies and demands, respectively. While for static flows,
the case of multiple sources and sinks can be easily reduced to the special case of a single source-sink
pair, the situation seems considerably more complicated for flows over time. Klinz~\cite{Klinz1994}
observed that the question whether the supplies and demands can be satisfied within a given time
horizon~$\theta$
(such a time horizon is then called \emph{feasible}) boils down to deciding whether a
particular submodular function~$d^{\theta}$ on the subsets of \emph{terminals} (sources and sinks) is non-negative.
This submodular function is induced by the submodular cut function of the corresponding time-expanded
network and can be evaluated via one maximum flow over time computation in~$O(\MCF)$ time.

Hoppe and Tardos~\cite{Hoppe2000} present the first (strongly) polynomial time algorithm for the Quickest
Transshipment Problem. In a first step, their algorithm determines the minimum feasible time
horizon~$\theta^*$ by solving a parametric submodular function minimization problem using Megiddo's
parametric search framework~\cite{Megiddo1979}. In a second step, they find an integral transshipment over time with time
horizon~$\theta^*$ via a reduction to a lex-max flow over time problem, for which they present a strongly
polynomial time algorithm. This reduction again relies on at most~$2k-2$ parametric submodular function
minimizations where~$k$ is the number of terminals. Schlöter and
Skutella~\cite{Schloeter2017,schloter2018flows} come up with an alternative way of computing a (not
necessarily integral) transshipment over time with given time horizon~$\theta^*$ that only requires one
submodular function minimization using Orlin's algorithm~\cite{Orlin2009} (or any another submodular
function minimization algorithm relying on Cunningham's framework~\cite{Cunningham1985}).

\paragraph{Evacuation Problem.}
The special case of the Quickest Transshipment Problem with a single sink is also called \emph{Evacuation Problem}.
For this special case (as well as for the symmetric setting with a single source),
Kamiyama~\cite{kamiyama2019} and Schlöter~\cite{schloter2018flows} independently present a faster
algorithm with running time \mbox{$\tilde{O}(m^2k^4+m^2nk)$}\footnote{We often use the
$\tilde{O}$-notation for running times, omitting all poly-logarithmic terms.} for determining the minimum
feasible time horizon~$\theta^*$ that does not rely on Megiddo's parametric search framework; see Table~\ref{Tbl:Running_times_QTP}.
\begin{table}[t]
\caption[Running]{Running times\footnotemark[\value{footnote}] of strongly polynomial quickest transshipment algorithms for networks with~$n$ nodes, $m$ arcs, 
and~$k$ terminals}
\label{Tbl:Running_times_QTP}
\centering
\small
\renewcommand{\arraystretch}{1.1}
\begin{tabular}{ll}
\hline
\multicolumn{2}{l}{\em Quickest Flow Problem (single source and single sink)}\\ 
\quad Saho \& Shigeno~\cite{Saho2017} & $\tilde O(m^2n)$\\
\hline
\multicolumn{2}{l}{\em Evacuation Problem (single source or single sink)}\\ 
\quad Kamiyama~\cite{kamiyama2019}, Schlöter~\cite{schloter2018flows} & $\tilde O(m^2k^5+m^2nk)$\\
\hline
\multicolumn{2}{l}{\em Quickest Transshipment Problem (multiple sources and sinks)}\\ 
\quad Hoppe \& Tardos \cite{Hoppe2000}& $\tilde O(m^4k^{15})$\\
\quad Schlöter \& Skutella \cite{Schloeter2017,schloter2018flows}\qquad${}$ & $\tilde O(m^4k^{14})$\\
\quad 
this paper & $\tilde O(m^2k^5+m^3k^3+m^3n)$\\
\hline
\end{tabular}
\end{table}
We describe this
algorithm in Section~\ref{sec:evac}. The corresponding transshipment can then be computed by an
additional submodular function minimization using Orlin's algorithm~\cite{Orlin2009} (or any another
submodular function minimization algorithm relying on Cunningham's framework), resulting in an overall
running time of $\tilde{O}(m^2k^5+m^2nk)$.

\paragraph{Our Contribution.}
The running times of all previous strongly polynomial algorithms~\cite{Hoppe2000,Schloeter2017,schloter2018flows} for the Quickest Transshipment
Problem are heavily dominated by the running time
required to solve parametric submodular function minimization problems. For this purpose, Megiddo's
parametric search framework~\cite{Megiddo1979} requires a fully combinatorial submodular function
minimization algorithm, that is, an algorithm whose only elementary operations are additions, scalar
multiplications, and comparisons.
The fastest such algorithm known is due to Iwata and Orlin~\cite{iwata2009simple} and needs $\tilde{O}(m^2k^7)$ time to minimize the submodular function~$d^\theta$. Pairing parametric search with this algorithm to
determine the minimum feasible time horizon~$\theta^*$ requires
$\tilde{O}(m^4k^{14})$ time, since parametric search causes a squaring of the base running time.
After computing the minimum feasible time horizon~$\theta^*$, Hoppe and Tardos~\cite{Hoppe2000}
require~$2k-2$ further parametric submodular function minimizations to determine an integral
transshipment over time, while Schlöter and Skutella~\cite{Schloeter2017,schloter2018flows} determine a
(fractional) transshipment over time via only one submodular function minimization.

We present a novel strongly polynomial algorithm for the Quickest Transshipment Problem, that does not rely
on parametric search. Instead, we develop a subtle extension of the Discrete Newton Method (or
Dinkelbach's Algorithm~\cite{Dinkelbach1967}) for computing the minimum feasible time horizon~$\theta^*$. In
terms of its running time, our algorithm beats all previous algorithms for the Quickest Transshipment
Problem by several orders of magnitude. It is close to the best known running times for the considerably
simpler Quickest Flow and Evacuation Problem; see Table~\ref{Tbl:Running_times_QTP}.

\begin{theorem}\label{Thm:main}
For the Quickest Transshipment Problem, the minimum feasible time horizon~$\theta^*$ and a transshipment
over time with time horizon~$\theta^*$ can be computed in $\tilde{O}(m^2k^5+m^3k^3+m^3n)$ time.
\end{theorem}

The Discrete Newton Method has proved useful in deriving efficient algorithms for a number of problems in
combinatorial optimization, including, for example, the Line Search Problem in a submodular
polytope~\cite{Goemans2017}, the Parametric Global Minimum Cut Problem~\cite{AissiMcCormickQueyranne2020},
Budgeted Network Problems~\cite{McCormickOrioloPeis2014}, and the Evacuation
Problem~\cite{kamiyama2019,schloter2018flows}.
In our context, computing the minimum feasible time horizon~$\theta^*$ boils down to finding the smallest
zero of the parametric submodular function~$\theta\mapsto\min_Sd^{\theta}(S)$ where the minimum is taken
over all subsets~$S$ of terminals. Unfortunately, this function, as the pointwise minimum of piecewise
linear and convex functions, is generally neither convex nor concave (see Figure~\ref{fig:lower_envelope}
below).

For the special case of the Evacuation Problem, Kamiyama~\cite{kamiyama2019} and
Schlöter~\cite{schloter2018flows} find~$\theta^*$ by successively computing zeros of
functions~$\theta\mapsto d^{\theta}(S)$ for particular terminal subsets~$S$. This basic approach
can also be applied to the more general Quickest Transshipment Problem, as we discuss in
Section~\ref{sec:evac}. We can, however, only prove an exponential bound on its running time. In
Section~\ref{sec:main} we therefore present a refined algorithm that reaches~$\theta^*$ in strongly
polynomial time. The novel algorithmic idea is to allow for larger jumps
towards~$\theta^*$. These jumps are carefully chosen such that their number is bounded in terms of
the Quickest Transshipment Problem's combinatorial structure. The analysis of our algorithm builds on the
work of Radzik~\cite{Radzik1998} as well as on the more recent results on ring families by Goemans,
Gupta, and Jaillet~\cite{Goemans2017}.

%
%

%% file: 1_Preliminaries.tex
%
%
A \emph{flow over time network} $\N = (D = (V,A), u, \tau, S^+, S^-)$ consists of a directed graph $D =
(V,A)$ together with integral \emph{arc capacities} $u\in\Z_{\geq 0}^A$, integral \emph{arc transit
times} $\tau\in\Z_{\geq 0}^A$, and disjoint sets of sources $S^+ \subseteq V$ and sinks $S^- \subseteq
V$.
Throughout this paper we set $n:=|V|$, $m := |A|$, and $k := |\terminals|$.
\emph{Supplies} and \emph{demands} at the sources and sinks, respectively, are given by a
function $b:\terminals\!\!\rightarrow\Z$ with~$b(s) \geq 0$ for~$s \in S^+$, $b(t) \leq 0$
for~$t\in S^-$, and $\sum_{v\in\terminals} b(v) = 0$.
A flow over time in~$\N$ with time horizon~$\theta$ specifies for each arc~$a\in A$ and each point in
time $\theta'\in[0,\theta]$ the rate at which flow enters arc~$a$ at time~$\theta'$. 
At each point in time the flow value on an arc $a \in A$ is upper bounded by its capacity $u_a$.
Additionally, a flow over time has to respect flow conservation.
In the following, we discuss some important properties of flow over time problems considered in this paper. For a thorough
introduction to flows over time we refer to the survey~\cite{Skutella-Intro09}.

Given a flow over time network~$\N$, a supply/demand-function~$b$, and a time horizon~$\theta\geq 0$, the
\emph{Transshipment Over Time Problem} $(\N,b,\theta)$ asks for a flow
over time with time horizon~$\theta$ in~$\N$ that satisfies all supplies and demands.
If such a flow over time exists, we say that~$(\N,b,\theta)$ is \emph{feasible} and that~$\theta$ is a
\emph{feasible time horizon} for~$(\N,b)$.
The \emph{Quickest Transshipment Problem}~$(\N,b)$ is to find a transshipment over time with minimum
feasible time horizon~$\theta^*$.

The feasibility of~$(\N,b,\theta)$ can be characterized via the minimum of a particular submodular
function on the set of terminals~$\terminals$ which we introduce next.
For a subset of terminals $S\subseteq\terminals$ and~$\theta\geq 0$, let~$o^{\theta}(S)$ be the maximum
amount of flow that can be sent from the sources in~$S^+\cap S$ to the sinks in~$S^-\setminus S$ within
time horizon~$\theta$. Klinz~\cite{Klinz1994} observed that~$\theta$ is a feasible time horizon
for~$(\N,b)$ if and only if
\begin{align}
o^{\theta}(S)~\geq~b(S)~\coloneqq~\sum_{v\in S}b(v)\qquad\text{for all~$S\subseteq\terminals$.}
\label{eq:feasibility}
\end{align}
We set~$d^\theta(S)\coloneqq o^\theta(S)-b(S)$ for~$S\subseteq\terminals$. Condition~\eqref{eq:feasibility} is equivalent to~$d^\theta(S)\geq0$ for all~$S\subseteq\terminals$.
The functions~$S\mapsto o^{\theta}(S)$ and~$S\mapsto d^{\theta}(S)$ are
submodular~\cite{Hoppe2000}. In particular, the feasibility of time horizon~$\theta$ can be checked via one submodular function minimization, since it is equivalent to
the non-negativity of~$d^\theta$.

For~$\theta\geq0$ and~$S\subseteq\terminals$, the value~$o^{\theta}(S)$ and thus~$d^\theta(S)$ can be
computed as follows. Consider the extended network $\N^S$ that is obtained by adding a super source~$s$
and a super sink~$t$ to~$\N$ together with arcs of infinite capacity and zero transit time connecting~$s$
to each~$v\in S^+\cap S$ and connecting each~$v\in S^-\setminus S$ to~$t$. By construction, $o^\theta(S)$
is equal to the value of a maximum $s$-$t$-flow over time in~$\N^S$ with time horizon~$\theta$.
Thus,~$o^\theta(S)$ can be computed via one static min-cost flow computation using the algorithm of
Ford and Fulkerson~\cite{Ford1958,Ford1962}: Create a network~$\N^{S,\theta}$ from~$\N^S$ by
connecting~$t$ to~$s$ with an additional arc~$a'$ of infinite capacity and negative `transit
time'~$-\theta$. Then,~$-o^\theta(S)$ is equal to the minimum cost of a (static) circulation
in~$\N^{S,\theta}$, where arc costs are equal to transit times.

For each~$S\subseteq\terminals$, the function~$\theta\mapsto o^{\theta}(S)$ thus corresponds to a
parametric min-cost flow problem in the network~$\N^{S,\theta}$. In particular, it is piecewise linear,
non-decreasing, and convex. The same properties therefore hold for the related function~$\theta\mapsto
d^\theta(S)=o^\theta(S)-b(S)$.

\begin{obs}\label{Obs:Structure_of_Breakpoints}
Let~$S\subseteq\terminals$ and let~$\theta'$ be a breakpoint of the function~$\theta\mapsto
d^{\theta}(S)$. Then there is a vector~$\lambda\in\{-1,0,1\}^A$ with $\theta'=\sum_{a \in
A}\lambda_a\tau_a$.
\end{obs}

\begin{proof}
The parametric minimum-cost flow problem in network~$\N^{S,\theta}$ can be solved by the minimum cycle
cancelling algorithm: Start with the zero-flow, which is an optimum solution for~$\theta=0$.
Due to arc~$a'$ of cost~$-\theta$, new negative cycles emerge in the residual network when~$\theta$ is
increased. Each such cycle~$C$ must contain~$a'$; augmenting flow along~$C$ induces a breakpoint at
time~$\theta'=\sum_{a\in C\setminus\{a'\}}\tau_a$. As~$C$ lives in a residual network, it may contain
backward arcs~$\cev{a}$ of cost~$-\tau_a$, where~$a\in A$ is the corresponding forward arc.
\end{proof}

Notice that the left-hand derivative of the function~$\theta\mapsto d^{\theta}(S)$ at breakpoint~$\theta'$ is equal to the flow value on arc~$a'$ in a min-cost circulation~$x^{\theta'}$ in~$\N^{S,\theta'}$.
This flow value is equal to the minimum capacity of an
$s$-$t$-cut in the subnetwork of~$\N^{S}$ induced by the arcs in the support of this min-cost circulation.
%
%
%
We define $\mincut^{\theta'}(S) := x^{\theta'}_{a'}$.
Note that $\mincut^{\theta'}(S)$ is the left-hand derivative of $\theta \mapsto d^\theta(S)$ at time $\theta'$.

We define the function~$d:\R_{\geq0}\to\R$ by~$d(\theta)\coloneqq\min_{S\subseteq\terminals}d^\theta(S)$,
i.e., $d$ is the lower envelope of functions~$\theta\mapsto d^\theta(S)$, $S\subseteq\terminals$, and
therefore piecewise linear and non-decreasing; see Figure~\ref{fig:lower_envelope} for an illustration.
\begin{figure}[t]
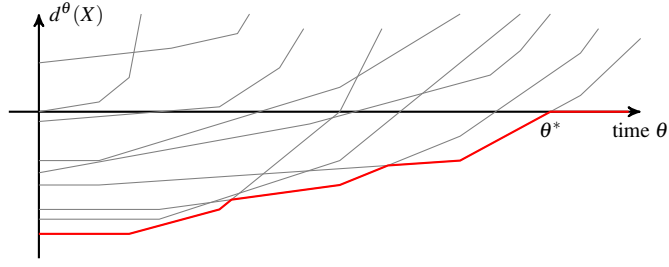

	\begin{center}
		\include{Figures/lower_envelope}
	\end{center}
	\caption{Each gray line corresponds to a subset of terminals~$S\subseteq\terminals$ and depicts the piecewise linear, increasing, and convex function $\theta \mapsto d^\theta(S)$;
	the lower envelope $\theta \mapsto d(\theta)$ is depicted in red.}
	\label{fig:lower_envelope}
\end{figure}
The minimum feasible time horizon is thus~$\theta^*=\min\{\theta\geq0\mid d(\theta)\geq0\}$. 

The flow over time model described above is also called the \emph{continuous time model} because it uses
a continuous time interval $[0,\infty)$.
When Ford and Fulkerson introduced flows over time in the 1950s~\cite{Ford1958,Ford1962}, they considered
a discrete time model.
In this model, discrete time steps $\{0,1,2,3,\ldots\}$ are used instead of the continuous interval
$[0,\infty)$. Fleischer and Tardos~\cite{Fleischer1998} show a close connection between the two models
for a number of flow over time problems, including the Quickest Transshipment Problem.
In this paper we exclusively work with the continuous time model. In particular, we do not require that
arc transit times or capacities are integral. Our strongly polynomial algorithm can solve instances where
transit times and capacities are arbitrary non-negative real numbers.

%% file: Figures/lower_envelope.tex
\begin{tikzpicture}[xscale = 0.8, yscale=0.65]
    \draw[->, >=stealth', thick] (-0.5,0)  -- (10,0) node(xline)[below]
        {\scriptsize time $\theta$};
    \draw[->, >=stealth', thick] (0, -3) -- (0,2) node(yline)[right] {$\scriptstyle d^\theta(X)$};
    \draw[thin, color = gray] (0,-1) -- (1,-1) -- (5,0.5)--(7,2);
    \draw[thin, color = gray] (0,-2.2)--(2,-2.2)--(5,-1)--(8,2);

    \draw[thin, color = gray] (0,-2.5) -- (1.5,-2.5)--(3,-2)--(4,-1)--(5,0)--(5.7,1.7);
    \draw[thin, color=gray] (0,-2)--(2,-2)--(5,-1.5)--(7,-0.5)--(9,1.2)--(9.3,1.7);
    \draw[thin, color=gray] (0,-1.5)--(1,-1.5)--(7,-1)--(9,1/3)--(10,1.5);
    
    \draw[thin, color=gray] (0,1)--(2.2,1.3)--(3.3,1.6)--(3.5,2);

    \draw[thin, color=gray] (0,0)--++(1,0.2)--++(0.5,0.5)--(1.7,2);

    \draw[thin, color=gray] (0,-1.25)--++(4.5,1)--++(3,1)--++(0.5,0.5)--(8.5,2);

    \draw[thin, color=gray] (0,-0.2)--++(3,0.3)--++(1,0.8)-- (4.4,1.7);

    \node[below] at (8.5,0) {$\scriptstyle\theta^*$};

    \coordinate (p1) at (3.2 ,-1.8);
    \coordinate (p2) at (5.8 ,-1.1);
    \draw[thick, color = red] (0,-2.5) -- (1.5,-2.5)--(3,-2)--(p1)--(5,-1.5)--(p2)--(7,-1)--(8.5,0)--(9.8,0);

\end{tikzpicture}

%% file: 2_Evacuation_Problem.tex
%
Given an instance $(\dynN,b)$ of the Quickest Transshipment Problem, a straightforward approach to
determine the minimum feasible time horizon of $(\dynN,b)$ is Algorithm~\ref{alg_evac}, which has first
been considered by Kamiyama~\cite{kamiyama2019} and Schlöter~\cite{schloter2018flows}. It computes a
sequence of infeasible time horizons~$0=\theta_0<\theta_1<\theta_2<\dots$ until it finds~$\theta^*$. The value~$\theta_{i+1}$ is always the zero of the function~$\theta\mapsto d^{\theta}(S_i)$ where~$S_i$
is a minimizer of the submodular function~$d^{\theta_i}$. The seemingly useless distinction
of~$\theta'_{i}$ and~$\theta_{i+1}$ in Lines~4 and 5 is owed to a crucial refinement introduced in Section~\ref{sec:main}
below.
%
%
\begin{algorithm}[t]
\DontPrintSemicolon
\caption{Discrete Newton like approach}
\label{alg_evac}
\KwIn{an instance $(\dynN,b)$ of the Quickest Transshipment Problem}
\KwOut{minimum feasible time horizon $\theta^*$ for $(\N,b)$}
$i:= 0$, $\;\theta_0 := 0$

\While{$d(\theta_i) < 0$}{
	$S_{i} := \operatorname{argmin}\{d^{\theta_{i}}(S)  :  S \subseteq S^+\cup S^- )\}$\;
	$\theta'_{i} := \min\{\theta: d^\theta(S_i) = 0\}$ \;
	$\theta_{i+1} := \theta'_i$\;
	$i := i+1$\;
}
\Return $\theta_{i}$
\end{algorithm}
In Figure~\ref{fig:alg_evac} three iterations of Algorithm~\ref{alg_evac} are illustrated.
\begin{figure}[h]
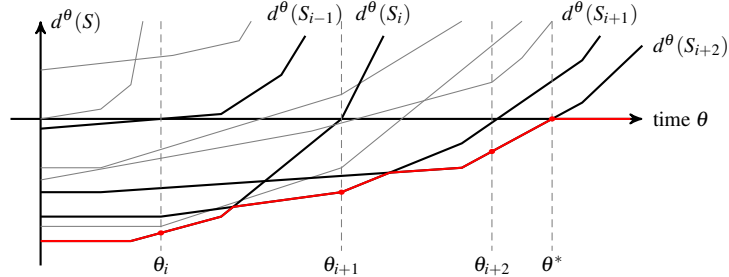

\begin{center}
\include{Figures/alg_evac}
\end{center}
\caption{Last three iterations of Algorithm~\ref{alg_evac}}
\label{fig:alg_evac}
\end{figure}

\begin{lemma}\label{Lem:Prop_of_Alg_evac}
Let $i \geq 0$ be an arbitrary iteration of Algorithm~\ref{alg_evac} and assume $b\not\equiv0$.
Then,\\
(i) iteration~$i$ takes $\tilde{O}(m^2k^3 + m^2n)$ time;\quad (ii) $\theta_{i}<\theta_{i+1}$;\quad (iii) $S_i \neq S_j$, for two distinct iterations $i,j\geq 0$.
\end{lemma}

\begin{proof} 
(i) In every iteration $i$, the submodular function minimizer $S_i$ can be computed in $\tilde{O}(m^2k^3)$
time with the algorithm of Lee, Sidford, and Wong~\cite{Lee2015}, while $\theta'_i$ can be determined in
$\tilde{O}(m^2n)$ time by solving a quickest $s$-$t$-flow problem in $\dynN^{S}$ with the algorithm of
Saho and Shigeno~\cite{Saho2017}.
For part~(ii), recall that $d^{\theta_{i}}(S_{i})<0$ (for~$i=0$ this follows from~$b\not\equiv0$) and $d^{\theta_{i+1}}(S_{i})=0$.
Since $d^\theta(S_{i})$ is non-decreasing in $\theta$, this implies that $\theta_{i}<\theta_{i+1}$.
For~(iii), take two distinct iterations~$i<j$ with~$\theta_i<\theta_{i+1}\leq\theta_j$ due to~(b). Since~$d^{\theta}(S_i)$ is non-decreasing in~$\theta$, we get
$d^{\theta_{i+1}}(S_i)=0\leq d^{\theta_j}(S_i)$. Thus~$d^{\theta_j}(S_j)<0\leq d^{\theta_j}(S_i)$ and~$S_i\neq S_j$.
\end{proof}

By Lemma~\ref{Lem:Prop_of_Alg_evac}~(iii), the number of iterations is at most~$2^k$ and Algorithm~\ref{alg_evac} runs in $\tilde{O}(2^k(m^2k^3+m^2n))$ time. 
The algorithm thus returns the smallest zero of the function~$d(\theta)$, that is, the minimum feasible time horizon~$\theta^*$ of~$(\dynN,b)$.
%
%
If~$|S^+|=1$ or~$|S^-|=1$, Algorithm~\ref{alg_evac} even runs in strongly polynomial
time~\cite{kamiyama2019,schloter2018flows}. We give a short sketch of the analysis for the Evacuation
Problem; the case~$|S^+|=1$ is symmetric.
%
%

If~$|S^-|=1$, the submodular functions $d^\theta$ and $d^{\theta'}$
are related by a so-called \emph{strong map} for all $0\leq\theta<\theta'$ \cite{Baumann2009}
(see~\cite{Topkis78} or~\cite{Nagano2007} for a formal definition of a strong map).
The strong map property in particular implies that the minimizers $S_0,\ldots,S_N$ computed during the
course of Algorithm~\ref{alg_evac} fulfill $S_N \subseteq S_{N-1} \subseteq \dots \subseteq S_0$,
provided $S_i$ is a minimal (or maximal) minimizer of~$d^{\theta_i}$ for all
$i=0,\dots,N$~\cite{Topkis78}.
Thus, by Lemma~\ref{Lem:Prop_of_Alg_evac} (iii), Algorithm~\ref{alg_evac} terminates after at most
$k$~iterations
%
%
and can be implemented to run in~$\tilde{O}(k(m^2k^3+m^2n))$ time.

Unfortunately, if $|S^-|>1$ and $|S^+|>1$, the submodular functions $d^\theta$ and $d^{\theta'}$ with
$0 \leq \theta < \theta'$ are in general \emph{not} related by a strong map, and the minimizers
$S_0,\ldots,S_N$ computed during the course of Algorithm~\ref{alg_evac} do not form a chain.
Thus, the analysis above cannot be extended to the Quickest Transshipment Problem.

%% file: Figures/alg_evac.tex
\begin{tikzpicture}[xscale = 0.8, yscale=0.65]
    \draw[->, >=stealth', thick] (-0.5,0)  -- (10,0) node(xline)[right]
        {\scriptsize time $\theta$};
    \draw[->, >=stealth', thick] (0, -3) -- (0,2) node(yline)[right] {$\scriptstyle d^\theta(S)$};

    \draw[thin, color = gray] (0,-1) -- (1,-1) -- (5,0.5)--(7,2);
    \draw[thin, color = gray] (0,-2.2)--(2,-2.2)--(5,-1)--(8,2);

    \draw[thick, color = black] (0,-2.5) -- (1.5,-2.5)--(3,-2)--(4,-1)--(5,0)--(5.7,1.7) node[above, black] {$\scriptstyle d^\theta(S_i)$} ;
    \draw[thick, color=black] (0,-2)--(2,-2)--(5,-1.5)--(7,-0.5)--(9,1.2)--(9.3,1.7)node[black, above] {$\scriptstyle d^\theta(S_{i+1})$};
    \draw[thick, color=black] (0,-1.5)--(1,-1.5)--(7,-1)--(9,1/3)--(10,1.5)node[black, right] {$\scriptstyle d^\theta(S_{i+2})$};
    
    \draw[thin, color=gray] (0,1)--(2.2,1.3)--(3.3,1.6)--(3.5,2);

    \draw[thin, color=gray] (0,0)--++(1,0.2)--++(0.5,0.5)--(1.7,2);

    \draw[thin, color=gray] (0,-1.25)--++(4.5,1)--++(3,1)--++(0.5,0.5)--(8.5,2);

    \draw[thick, color=black] (0,-0.2)--++(3,0.3)--++(1,0.8)-- (4.4,1.7) node[black, above] {$\scriptstyle d^\theta(S_{i-1})$};

    \draw[densely dashed, gray, name path = it1] (2,2) -- (2,-2.7);
    \node[below] at (2,-2.6){$\scriptstyle\theta_{i}$};

    \draw[densely dashed, gray, name path = it2] (5,2) -- (5,-2.7);
    \node[below] at (5,-2.6){$\scriptstyle\theta_{i+1}$};

    \draw[densely dashed, gray, name path = it3] (7.5,2) -- (7.5,-2.7);
    \node[below] at (7.5,-2.6){$\scriptstyle\theta_{i+2}$};
    \node[below] at (8.5,-2.6){$\scriptstyle\theta^*$};

    \draw[densely dashed, gray, name path = it4] (8.5,2) -- (8.5,-2.7);

    \node (m1) at (2,-2.33333){};
    \node (m2) at (5,-1.5){};
    \node (m3) at (7.5,-0.666666){};

    \node (m4) at (8.5,0){};

    \fill[red] (m1) circle (1.5pt);
    \fill[red] (m2) circle (1.5pt);
    \fill[red] (m3) circle (1.5pt);
    \fill[red] (m4) circle (1.5pt);

    \coordinate (p1) at (3.2 ,-1.8);
    \coordinate (p2) at (5.8 ,-1.1);
    \draw[thick, color = red] (0,-2.5) -- (1.5,-2.5)--(3,-2)--(p1)--(5,-1.5)--(p2)--(7,-1)--(8.5,0)--(9.8,0);
\end{tikzpicture}

%% file: 3_Main_Alg.tex
%
In this section we present an extension of Algorithm~\ref{alg_evac} that determines the minimum feasible
time horizon for a Quickest Transshipment Problem and has strongly polynomial running time.
As Algorithm~\ref{alg_evac}, the refined algorithm computes a sequence $0 = \theta_0 < \theta_1 < \theta_2 < \ldots < \ldots$ of infeasible time horizons until it finds $\theta^*$.
Every iteration $i$ of the while-loop in the refined Algorithm~\ref{alg_main} consists of \emph{two} updates.
At first the current time horizon~$\theta_i$ is updated to~$\theta_i'$, exactly as in
Algorithm~\ref{alg_evac}.
In particular, parts~$(ii)$ and~$(iii)$ of Lemma~\ref{Lem:Prop_of_Alg_evac} are also valid for
Algorithm~\ref{alg_main}, which implies that the algorithm terminates in finite time and correctly
returns the minimum feasible time horizon~$\theta^*$ for a given Quickest Transshipment Problem~$(\N,b)$.
The second update (jumping from~$\theta_i'$ to~$\theta_{i+1}$) ensures that Algorithm~\ref{alg_main} terminates after a strongly polynomial number of iterations.
See Figure~\ref{fig:alg_main} for a visualization of an iteration of Algorithm~\ref{alg_main}.
\begin{algorithm}[b]
\DontPrintSemicolon
\caption{Generalized Discrete Newton Algorithm}
\label{alg_main}
\KwIn{a dynamic network $\N$ and a supply/demand vector $b$}
\KwOut{minimum feasible time horizon $\theta^*$ for $(\N,b)$}
$i:= 0$, $\;\theta_0 := 0$, 
$\;J:=\bigl\{2^0,2^1,2^2,\ldots,2^{\lceil \log_2(k^2/4)\rceil}\bigr\}$\; 
\While{$d(\theta_i) < 0$}{
	$S_{i} := \operatorname{argmin}\{d^{\theta_{i}}(S)  :  S \subseteq S^+ \cup S^-)\}$\;
	$\theta_i' := \min\{\theta: d^\theta(S_i) = 0\}$ \;
	$\theta_{i+1} := \max\Bigl(  \{\theta'_i\} \cup \Bigl\{\theta =\theta_i' + j\cdot \frac{-d(\theta_i')}{\mincut ^{\theta_i'}(S_i)}: d(\theta)<0,~j \in J\Bigr\}\Bigr)$\;
	$i := i+1$\;
}
\Return $\theta^* := \theta_{i}$
\end{algorithm}
\begin{theorem}\label{Thm:Running_Time_Main}
Algorithm~\ref{alg_main} can be implemented to run in $\tilde O(m^2k^5+m^3k^3+m^3n)$ time and it returns the minimum feasible time horizon $\theta^*$ for a given Quickest Transshipment Problem $(\N,b)$.
\end{theorem}

Before we precisely analyze Algorithm~\ref{alg_main}, we discuss a high-level explanation why the algorithm terminates in strongly polynomial time.
Since it is easy to see that each iteration of the algorithm runs in strongly polynomial time, the main effort is in bounding the number of iterations.
Let $I=\{0,\ldots,N\}$ be the iterations of the while-loop in Algorithm~\ref{alg_main}.
If the maximum computed in Line~5 of iteration~$i \in I$ is equal to $\theta_i'$, we set $j_i:=0$.
Otherwise, the maximum computed in Line~5 of iteration $i \in I$ is equal to $\theta_i' + j \cdot \frac{-d(\theta_i')}{\cut^{\theta_i'}(S_i)}$ with $j \in J$ and we set $j_i := j$.
First, we consider the iterations in which the algorithm makes the longest possible jump during the second update:
\[
\textstyle I^1~:=~\bigl\{i \in I: j_i =  2^{\lceil \log_2(k^2/4)\rceil}\bigr\}.
\]
In Section~\ref{subsec:I1} we prove that the jump in an iteration~$i\in I^1$ is so long (i.e., the
increase in the time horizon is so large) that we can identify a source-sink pair that could fulfill the
entirety of supplies and demands in the networks on its own.
Thus, this pair cannot be `encountered' in later iterations, yielding a bound of~$k^2/4$ on~$|I^1|$
(see Lemma~\ref{Lem:Cardinality_I_1}).

For all other iterations, the fact that we did not make the largest possible jump guarantees that we have
at least halved the distance to the minimum feasible time horizon (see Lemma~\ref{Lem:alg_main_distance}).
In Section~\ref{subsec:I2} we bound the number of iterations that `move across a breakpoint':
\[
I^2~:=~\bigl\{i \in I \setminus I^1: \theta\mapsto d^\theta(S) \text{ has a breakpoint in $[\theta_i, \theta_{i+1}]$ for some~$S\subseteq\terminals$}\bigr\}.
\]
The structure of breakpoints stated in Observation~\ref{Obs:Structure_of_Breakpoints} together with a
lemma due to Goemans implies that~$|I^2|$ is at most~$O(m\log m)$ (see Lemma~\ref{Lem:Cardinality_I_2}).
It remains to consider the remaining iterations
\[
I^3~:=~I\setminus(I^1 \cup I^2).	
\]
In this case all curves~$\theta\mapsto
d^\theta(S)$, for~$S\subseteq\terminals$, are linear `within' these iterations.
This removes the principal hurdle for using the classical Discrete Newton Method. 
We use ideas from the analysis of a Discrete Newton Method for line search in a polymatroid by Goemans et
al.~\cite{Goemans2017}, which we introduce in Section~\ref{subsec:ring}.
The key technique is to extract a chain of ring families from the sets~$S_i$ with~$i\in I^3$.
The length of such chains is bounded (Theorem~\ref{Thm:Chain_of_ring_families}). 
Based upon this, we prove in Section~\ref{subsec:I3} that~$|I^3|\in O(k^2\log k+m\log m\log k)$ (see
Lemma~\ref{Lem:Cardinality_I_3}).
Finally, in Section~\ref{subsec:theorem} we prove Theorem~\ref{Thm:Running_Time_Main}, which implies Theorem~\ref{Thm:main}.

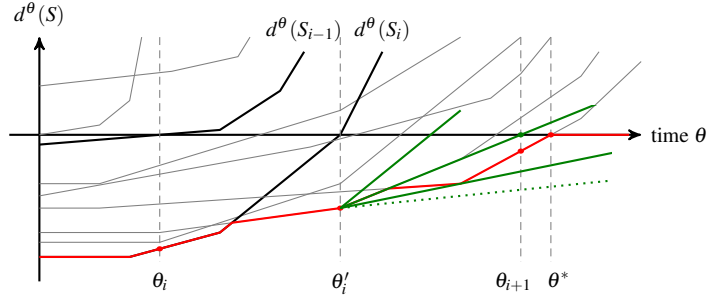
\begin{figure}[t]
\begin{center}
\input{Figures/alg_main}
\end{center}
\caption{One iteration of Algorithm~\ref{alg_main}.
The green rays correspond to $j=2^0,2^1,2^2,\ldots$ and their 
slopes are equal to $\cut^{\theta_i'}(S_i)$, $\cut^{\theta_i'}(S_i)/2$,
$\cut^{\theta_i'}(S_i)/4,\ldots$, respectively.
As the ray with slope $\cut^{\theta_i'}(S_i)/4$ corresponding to $j = 4$ intersects the
time-axis after time $\theta^*$, the value $j=2$ is chosen in the displayed iteration $i$.
}
\label{fig:alg_main}
\end{figure}

\subsection{Bounding the Number of Iterations in~$I^1$}
\label{subsec:I1}

\begin{lemma}\label{Lem:Cardinality_I_1}
$|I^1|~\leq~k^2/4$.
\end{lemma}
\begin{proof}

Consider an iteration $i\in I^1$, i.e., $j_{i}=2^{\left\lceil \log_2(k^2/4)\right\rceil}\geq k^2/4$.
Let $x$ be a minimum-cost circulation in~$\dynN^{S_i,\theta'_i}$ with transit times as costs.
The static flow~$x|_{\dynN}$ induced by~$x$ on~$\dynN$ sends flow from~$S^+\cap S_i$ to~$S^-\setminus S_i$ of total value~$x_{a'}$, where~$a'$ is the arc connecting super-sink~$t$ to super-source~$s$ in~$\dynN^{S_i,\theta'_i}$.
Let $\cP_i$ denote the set of all paths from~$S^+\cap S_i$ to~$S^-\setminus S_i$ in~$\dynN$, and let $(x_P)_{P \in \cP_i}$ be a path decomposition of~$x|_{\dynN}$.
Recall that~$\mincut^{\theta_i'}(S_i)=x_{a'}=\sum_{P\in\cP_i}x_P$.
Since there are at most $k^2/4$ many source-sink pairs, there has to exist at least one source-sink pair~$s',t'$ with~$s'\in S_i$ and~$t'\in S^-\setminus S_i$ such that
\begin{align}\label{eq:lemma_big_jump_1}
\frac{4}{k^2} \cdot  \mincut^{\theta'_i}(S_i)~
\leq~\sum_{\substack{P\in\cP_i,\\ \text{$P$ is $s'$-$t'$-path}}} x_P~
\leq~\cut^{\theta_i'}(S)
\end{align}
for all~$S\subseteq\terminals$ with~$s'\in S\cap S^+$ and~$t'\in S^-\setminus S$.
We claim that, for any~$\ell>i$, it is impossible that both~$s'\in S_{\ell}\cap S^+$
and~$t'\in S^-\setminus S_{\ell}$.
Assume by contradiction that~$s'\in S_{\ell}$ and~$t'\in S^-\setminus S_{\ell}$ for some~$\ell>i$. 
The left-hand derivative of the convex function~$\theta\mapsto d^\theta(S_\ell)$ at time~$\theta_i'$ is equal to~$\cut^{\theta_i'}(S_\ell)$ and therefore at least~$\frac4{k^2}\cut^{\theta_i'}(S_i)$ by~\eqref{eq:lemma_big_jump_1}. Thus,
\begin{align}\label{eq::lemma_big_jump_2}
\begin{split}
d^{\theta_{\ell}}(S_{\ell})~
&\geq~ d^{\theta_{i}'}(S_{\ell}) + (\theta_{\ell}-\theta_i') \cdot \tfrac4{k^2}\cdot \mincut^{\theta_i'}(S_i)\\
&\geq~d(\theta_i') + (\theta_{i+1}-\theta_i') \cdot \tfrac4{k^2}\cdot \mincut^{\theta_i'}(S_i)~
=~d(\theta_i') + j_i\cdot\tfrac{-d(\theta_i')}{\mincut^{\theta_i'}(S_i)} \cdot \tfrac4{k^2}\cdot \mincut^{\theta_i'}(S_i)~\geq~0,
\end{split}
\end{align}
where the second inequality follows from $d(\theta_i')=\min_{S}d^{\theta_{i}'}(S)$
and~$\theta_\ell\geq\theta_{i+1}$ (Lemma~\ref{Lem:Prop_of_Alg_evac}~(ii)), and the last inequality follows
from~$j_i\geq k^2/4$.
Notice that~\eqref{eq::lemma_big_jump_2} yields a contradiction since~$d^{\theta_{\ell}}(S_{\ell})<0$.

We have thus shown that, after iteration~$i$, the source-sink pair $s',t'$ cannot `occur' anymore.
That is, in each iteration~$i\in I^1$ we lose one source-sink pair.
We conclude that~$|I^1|\leq k^2/4$, as there are at most~$k^2/4$ many distinct source-sink pairs, 
\end{proof}


%

\subsection{Bounding the Number of Iterations in~$I^2$}
\label{subsec:I2}

In order to bound~$|I^2|$, we use the following corollary of a result by Goemans communicated by Radzik~\cite{Radzik1998}.

\begin{lemma}\label{Lem:Radzik}
Let $c \in \R^{p}_{\geq 0}$ and let $\lambda^1,\ldots,\lambda^q \in \{-2,-1,0,1,2\}^p$. If 
$$0 < c^T \lambda^{i+1} \leq \tfrac12 c^T \lambda^i \qquad \text{for all $i \in \{1,\ldots,q-1\}$,}$$
then $q \in O(p\log p)$.
\end{lemma}

\begin{proof}
Goemans~\cite{Radzik1998} proved the special case with~$\lambda^1,\ldots,\lambda^q \in \{-1,0,1\}^p$.
Lemma~\ref{Lem:Radzik} follows by applying this special case to $c':=\bigl({c\atop c}\bigr)\in
\R^{2p}_{\geq 0}$.
\end{proof}

In order to be able to apply Lemma~\ref{Lem:Radzik} in our setting, we make the following observation:
In all iterations in~$I^2$ and~$I^3$, the second update (updating $\theta'_{i}$ to $\theta_{i+1}$)
reduces the distance of the current time $\theta_i$ to the minimal feasible time horizon $\theta^*$ by at
least half.
\begin{lemma}\label{Lem:alg_main_distance}
For every  $i \in I^2 \cup I^3$ with $i<N$ we have $\tfrac{1}{2} (\theta^*- \theta_{i}) \leq \theta_{i+1}-\theta_i$.
If $i \in I^2\cup I^3$, $i<N$, and $j_i > 0$, then $ \frac12(\theta^*-\theta_i') \leq \theta_{i+1}-\theta_i'$.
\end{lemma}

\begin{proof}
	We start by proving the first statement of the lemma.
	For an iteration $i \in I^2 \cup I^3$ with $i<N$ we distinguish the cases $j_i>0$ and $j_i = 0$.
	We first consider the case $j_i > 0$, i.e., $\theta_{i+1} = \theta'_{i} + j_i \cdot \frac{-d(\theta_i')}{\cut^{\theta_i'}(S_i)}$ and the time $\eta := \theta'_i + 2\cdot j_i \cdot \frac{-d(\theta_i')}{\mincut ^{\theta_i'}(S_i)}$ fulfills $d(\eta) \geq 0$ and thus $\eta \geq \theta^*$ by the choice of $j_i$ in Line~5 of Algorithm~\ref{alg_main}.
	Overall, we have
	\begin{align}\label{eq:dist_j_is_not_zero}
	\tfrac{1}{2}(\theta^*-\theta_i)
	~\leq~ \tfrac{1}{2}(\eta - \theta_i')+\tfrac{1}{2}(\theta_i' - \theta_{i})
	~=~ \theta_{i+1} - \theta_i' +  \tfrac{1}{2}(\theta_i' - \theta_{i})
	~<~\theta_{i+1}- \theta_i.
	\end{align}
	Next we consider the case $j_i =0$; an illustration is given in Figure~\ref{Fig:j_is_zero}.
	\begin{figure}[ht!]
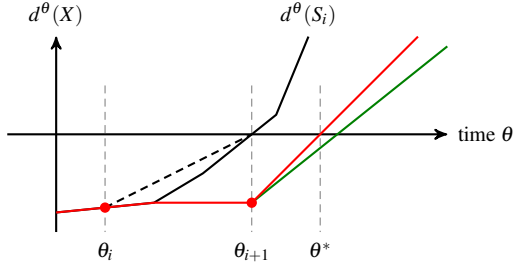

	\begin{center}
	\include{Figures/case_j_is_zero}
	\end{center}
	\caption{An iteration~$i$ of Algorithm~\ref{alg_main} with $j_i=0$.
	The red line shows the lower envelope~$\theta\mapsto d(\theta)$.
	The green ray corresponds to~$j_i=1$.
	Its slope is $\cut^{\theta_{i+1}}(S_i)$ and it intersects the time-axis after~$\theta^*$.
	The dashed line between $(\theta_i,d(\theta_i))$ and $(\theta_{i+1},0)$ has a slope that is smaller than $\cut^{\theta_{i+1}}(S_i)$.
	}
	\label{Fig:j_is_zero}
	\end{figure}
	In this case we have $\theta_{i+1} = \theta_i'$, and $\eta := \theta_{i+1} + \frac{-d(\theta_{i+1})}{\mincut^{\theta_{i+1}}(S_i)}\geq\theta^*$.

	Let $r:= -d(\theta_i)/(\theta_{i+1}-\theta_i)$ be the slope of the line through points $(\theta_i,d(\theta_i))$ and $(\theta_{i+1},0)$.
	Since~$\theta \mapsto o^\theta(S_i)$ is convex, we get~$r\leq\mincut^{\theta_{i+1}}(S_i)$.\
	Together with $-d(\theta_i) \geq -d(\theta_{i+1})$ for the last inequality, this yields
	\begin{align*}
	\begin{split}
	\tfrac{1}{2}(\theta^*-\theta_i)~
	\leq~\tfrac{1}{2}(\eta - \theta_{i+1} + \theta_{i+1}-\theta_i)~ 
	=~\frac{1}{2}\left(\frac{-d(\theta_{i+1})}{\cut^{\theta_{i+1}}(S_i)} + \frac{-d(\theta_i)}{r} \right)
	\leq~\frac{-d(\theta_i)}{r}~
	=~\theta_{i+1}-\theta_i.	
	\end{split}
	\end{align*}

	The second part of the lemma is directly implied by~\eqref{eq:dist_j_is_not_zero}.
\end{proof}

\begin{lemma}\label{Lem:Cardinality_I_2}
$|I^2|~\in~O(m\log m)$.
\end{lemma}

\begin{proof}
Let $I^2=\{i_1,i_2,\cdots,i_{|I^2|}\}$ with $i_1<i_2<\cdots<i_{|I^2|}$. 
By the definition of $I^2$ there exists a breakpoint $r_\ell\in\left[\theta_{i_\ell},\theta_{i_{\ell}+1}\right]$ for each $1\leq \ell \leq |I^2|$.
Moreover, let~$r_{0}$ be the largest infeasible breakpoint, i.e.,
\[
r_{0}~:=~\max\{\theta'\in [0,\theta^*]:\theta\mapsto d^\theta(S) \text{ has a breakpoint at $\theta'$ for some~$S\in\terminals$}\}.
\]
By Observation~\ref{Obs:Structure_of_Breakpoints}, there exist vectors $\lambda^\ell\in\{-1,0,1\}^A$ such that $r_\ell = \sum_{a \in A} \lambda^\ell_a \tau_a$ for each $0\leq \ell \leq |I^2|$.
For two consecutive iterations $i_\ell$ and $i_{\ell+1}$ with $1 \leq \ell \leq |I^2|-1$,
Lemma~\ref{Lem:alg_main_distance} yields
\[
(\theta_{i_{\ell+1}}-\theta_{i_\ell}) 
~\geq~ \tfrac{1}{2}(\theta^* - \theta_{i_\ell}) 
~\geq~ \tfrac{1}{2}(r_0 - \theta_{i_\ell}),
\]
and thus $(r_0-\theta_{i_{\ell+1}}) \leq \tfrac{1}{2}(r_0-\theta_{i_\ell})$.
For two iterations $i_\ell$ and $i_{\ell+2}$ with $1 \leq \ell \leq |I^2|-2$ this implies
\begin{align*}
r_0 - r_{\ell + 2} 
~\leq~ r_0 - \theta_{i_{\ell+2}}  
~\leq~ \tfrac{1}{2}(r_0- \theta_{i_{\ell+1}}) 
~\leq~ \tfrac{1}{2}(r_0 - r_\ell).
\end{align*}
%
Since $r_0-r_\ell = \sum_{a \in A}(\lambda_a^0 - \lambda_a^\ell) \tau_a = \sum_{a \in A}\mu_a \tau_a$ with $\mu_a \in
\{-2,-1,0,1,2\}$ for every~$a \in A$, we can apply Lemma~\ref{Lem:Radzik} to every second element in
$\{r_0-r_1,r_0-r_2,\dots,r_0-r_{|I^2|}\}$, which yields $|I^2|\in O(m\log m)$. 
\end{proof}


\subsection{Ring Families}
\label{subsec:ring}

In order to derive a bound on~$|I^3|$ below, we use the concept of ring families.
%
A \emph{ring family}~$\mathcal L\subseteq 2^W$ over a ground set~$W$ is a family of subsets closed under
taking unions and intersections. For an arbitrary family of subsets~$\cT\subseteq 2^W$, let~$\cR(\cT)$ be
the smallest ring family containing~$\cT$.
For our purposes,~$W$ is always the set of terminals~$\terminals$. The cardinality of~$W$ is therefore
denoted by~$k$.
The following bound on the length of a chain of ring families from~\cite{Goemans2017} is crucial for our
analysis.

\begin{theorem}[\cite{Goemans2017}]\label{Thm:Chain_of_ring_families}
For a finite set~$W$ with $k:=|W|$, consider a chain of ring families $\mathcal{L}_0 = \emptyset \neq \mathcal{L}_1 \subsetneq \mathcal{L}_2 \subsetneq \ldots \subsetneq \mathcal{L}_\lambda\subseteq 2^W$.
Then $\lambda \leq {k+1 \choose 2}+1$.
\end{theorem}

We will also make use of the following three, more technical results from~\cite{Goemans2017}. The first result bounds the value of a submodular function on a ring family~$\cR(\cT)$ under the assumption that it is non-positive on~$\cT$.

\begin{lemma}[\cite{Goemans2017}]\label{Lem:Value_on_ring_families}
Let $\cT \subseteq 2^W$ for some finite set $W$ with $k := |W|$ and let $f:2^W \to \R$ be a submodular function with minimum value $f_{\min} = \min_{S \subseteq W} f(S)$.
If $f(S) \leq 0$ for all~$S\in\cT$, 
then $f(S) \leq \frac{k^2}{4}|f_{\min}|$ for each~$S \in \cR(\cT)$.
\end{lemma}

The second result bounds the value of a submodular function over an expansion of a ring family.

\begin{lemma}[\cite{Goemans2017}]\label{Lem:max_ring_family}
Let $f:2^W \rightarrow \R$ be a submodular function with minimum value $f_{\min} = \min_{S \subseteq W} f(S) \leq 0$.
Let $\cL$ be a ring family over $W$ and $T \not \in \cL$.
Define $\cL' := \cR(\cL \cup\{T\})$, $M := \max_{S \in \cL} f(S)$, and $M' := \max_{S \in \cL'} f(S)$.
Then, $M' \leq 2(M+|f_{\min}|) + f(T)$.
\end{lemma}

The third result can be used to construct a chain of ring families.

\begin{lemma}[\cite{Goemans2017}]\label{lem:goe}
Let~$f:2^W\to\R$ be a submodular function with minimum value $f_{\min} =
\min_{S \subseteq W} f(S) \leq 0$. Consider a sequence of distinct sets $T_1,T_2,\dots,T_q$ such that
$f(T_1)=f_{\min}$, $f(T_2)>2|f_{\min}|$, and $f(T_i)\geq 4f(T_{i-1})$ for $3\leq i\leq q$. Then
$T_i\not\in\cR(\{T_1,\dots,T_{i-1}\})$ for all $1<i\leq q$. 
\end{lemma}

In order to apply these three lemmas in our analysis, we need to first prove certain properties of the
subsets~$S_i$ found by Algorithm~\ref{alg_main}. The following is an adaptation
of~\cite[Lemma~4]{Goemans2017}.

\begin{lemma}\label{Lem:Int_new}
Let $[u,v]$ be an interval of consecutive iterations in $I^3$.
Then we have\quad (i)~$d^{\theta_{v}}(S_{v}) = d(\theta_{v})$,\quad (ii)~$d^{\theta_{v}}(S_{{v-1}}) > 0$,\quad (iii) $d^{\theta_{v}}(S_{{v-2}}) > |d(\theta_{v})|$,\quad (iv) $d^{\theta_{v}}(S_{{i-1}})>2d^{\theta_{v}}(S_{i})$ for $i \in [u+1, v-1]$.
\end{lemma}
In order to prove Lemma~\ref{Lem:Int_new}, we need the following result on the slopes of curves found by Algorithm~\ref{alg_main} during iterations in~$I^3$.

\begin{lemma}\label{Lem:prop_of_alg_main_I_3}
Let $[u,v]=\{u,u+1,\ldots,v\} \subseteq I^3$ be an interval of consecutive iterations in $I^3$.
Then, $j_i>0$ and $\mincut^{\theta_{i}}(S_{i}) > \mincut^{\theta_{i+1}}(S_{i+1})$
for all $i \in [u,v-1]$.
\end{lemma}

\begin{proof}
Consider an arbitrary $i \in [u,v-1]$ and let $S'_i$ be a minimizer of $d(\theta_i')$. 
Note that $d^{\theta_i}(S_i')<0$.
Recall that $\mincut^{\theta'}(S)$ is equal to the left-hand derivative of $\theta \mapsto d^{\theta}(S)$ at time~$\theta'$ for~$S\subseteq\terminals$.
By definition of $I^3$, the functions $\theta \mapsto \mincut^{\theta}(S_i)$ and  $\theta \mapsto \mincut^{\theta}(S_i')$ are constant over $[\theta_u,\theta_{v+1}]$.
Thus,
\begin{align}\label{eq:prop_of_alg_main_I_3_1}
	0~=~d^{\theta_i'}(S_i)~=~d^{\theta_i}(S_i) + (\theta_i'- \theta_i) \cdot \mincut^{\theta_i}(S_i), 
\end{align}
and, since $d^{\theta_i}(S_i) \leq d^{\theta_i}(S_i')$, 
\begin{align}\label{eq:prop_of_alg_main_I_3_2}
	0~>~d^{\theta_i'}(S_i')~\geq~d^{\theta_i}(S_i) + (\theta_i' - \theta_i)\cdot \mincut^{\theta'_i}(S'_i).
\end{align}
Combining~\eqref{eq:prop_of_alg_main_I_3_1} and~\eqref{eq:prop_of_alg_main_I_3_2} yields that $\mincut^{\theta_{i}}(S_{i})>\mincut^{\theta'_{{i}}}(S_{{i}}')$.
By similar arguments we get~$\mincut^{\theta_{i}'}(S'_i)\geq\mincut^{\theta_{i+1}}(S_{i+1})$.

Assume by contradiction that $j_i = 0$.
Then we have  $\theta'_{i} + \frac{-d(\theta_i')}{\mincut^{\theta_i'}(S_i)}\geq\theta^*$, and the fact that $\theta \mapsto \mincut^{\theta}(S_i')$ is constant during $[\theta_u, \theta_{v+1}]$ implies $d^\eta(S_i') = 0$ with $\eta := \theta_i' + \frac{-d(\theta_i')}{\mincut^{\theta_i'}(S'_i)}$.
Since $\mincut^{\theta_i}(S_i) > \mincut^{\theta'_i}(S'_i)$, we get
\begin{align*}
\theta^* 
~\leq~ \theta_i' + \frac{-d(\theta_i')}{\mincut^{\theta_i'}(S_i)} 
~<~ \theta_i' + \frac{-d(\theta_i')}{\mincut^{\theta_i'}(S_i')} 
~=~ \eta  
~\leq~ \theta^*,
\end{align*}
a contradiction.
\end{proof}

\begin{proof}[Proof of Lemma~\ref{Lem:Int_new}]
%
We start with (iv):
Let $S_i'$ be a minimizer of $d^{\theta_i'}$ for all $i \in I$ and 	let $i \in [u+1,v-1]$.
Since the function $\theta \mapsto \mincut^{\theta}(S_{i-1})$ is constant for $\theta\in[\theta_u,\theta_{v+1}]$ and $d^{\theta_{i-1}'}(S_{i-1}) = 0$, we get 
\begin{align}\label{eq:Int_new}
d^{\theta_v}(S_{i-1}) 
~=~\mincut^{\theta_{i-1}}(S_{i-1})\cdot (\theta_{v} - \theta_{i-1}').
\end{align}
Since $j_{i-1} > 0$ by Lemma~\ref{Lem:prop_of_alg_main_I_3}, we can apply the second part of Lemma~\ref{Lem:alg_main_distance} and conclude that 
\[
\theta_i-\theta'_{i-1} 
~\geq~\tfrac12(\theta^*-\theta'_{i-1})
~\geq~\tfrac12(\theta_v-\theta'_{i-1}).
\]
This implies that
$\theta_v- \theta_{i}  \leq \frac12 (\theta_v-\theta_{i-1}')$, and with  Lemma~\ref{Lem:prop_of_alg_main_I_3} and~\eqref{eq:Int_new} we get
\begin{align*}
d^{\theta_v}(S_{i-1}) 
~=~\mincut^{\theta_{i-1}}(S_{i-1})\cdot (\theta_{v} - \theta_{i-1}')
~&>~2\mincut^{\theta_{i}}(S_i)\cdot (\theta_{v} - \theta_{i} )\\
&>~ 2\bigl(\mincut^{\theta_{i}}(S_{i})\cdot (\theta_{v} - \theta_{i} ) + d^{\theta_{i}}(S_{i})\bigr)
~=~ 2d^{\theta_{v}}(S_{i}),
\end{align*}
where the last inequality follows from~$d^{\theta_i}(S_i)<0$. This concludes the proof of~(iv).
Properties~(i) and~(ii) are clear.
By Lemma~\ref{Lem:prop_of_alg_main_I_3},
\begin{align*}
d^{\theta_{v}}(S_{v-2}) 
~&=~ \mincut^{\theta_{v-2}}(S_{v-2})\cdot (\theta_{v}-\theta_{v-1}) + d^{\theta{v-1}}(S_{v-2})\\
&>~ \mincut^{\theta_{v-1}}(S_{v-1})\cdot (\theta_{v}-\theta_{v-1})
~>~ |d^{\theta_{v-1}}(S_{v-1})| ~=~|d(\theta_{v-1})|
~\geq~ |d(\theta_v)|.
\end{align*}
This concludes the proof of~(iii).
\end{proof}

\subsection{Bounding the Number of Iterations in~$I^3$}
\label{subsec:I3}

With the help of the techniques from the previous section, we can finally prove a bound on~$|I^3|$. The
proof of the following lemma is inspired by~\cite[Proof of Theorem~6]{Goemans2017}.

\begin{lemma}\label{Lem:Cardinality_I_3}
$|I^3|~\in~O(k^2 \log k+m\log m\log k)$.
\end{lemma}

\begin{proof}
We partition~$I^3$ into maximal intervals (blocks) of consecutive iterations and order these blocks
decreasingly, i.e., $I^3 = \bigcup_{p = 1}^q[u_p, v_p]$ with $u_{p-1}> v_{p}+1$ for $1 < p
\leq q$.
Since every gap between two blocks contains at least one iteration in~$I^1\cup I^2$, we get~$q\leq
|I^1|+|I^2|$.
We construct a sequence of sets $T_1,T_2,\ldots$ such that no set in the sequence is in the ring closure
of its predecessors.

We start the construction with the sequence of sets~$\cT_1$ obtained by using Lemma~\ref{lem:goe} in
combination with Lemma~\ref{Lem:Int_new} applied to the first interval~$[u_1,v_1]$ with the submodular
function~$f=d^{\theta_{v_1}}$. That is, $\cT_1$ consists of the sets $S_{v_1}$ and then~$S_i$ for every other~$i \in [u_1,v_1-3]$.

We continue the construction of the sequence of sets inductively. 
Suppose that we have extracted a sequence of sets~$\cT_{p-1}$ from the intervals~$[u_s,v_s]$ for
$s<p$ such that no set is in the ring closure of its predecessors.
We show, that we can extend $\cT_{p-1}$ by adding every other set in $\bigl[u_p,\ldots,v_p - \lceil \log_2 k^2 \rceil\bigr]$, yielding the longer sequence~$\cT_p$.
To this end, consider the submodular function $f := d^{\eta}$ for time~$\eta := \theta_{v_p}$, and let
$f_{\min}<0$ be its minimum value.
By construction, $f(T)<0$ for all~$T\in \cT_{p-1}$.
Thus, Lemma~\ref{Lem:Value_on_ring_families} yields $f(S) \leq \frac{k^2}{4}|\fmin|$ for all $S \in \cR(\cT_{p-1})$.

The set~$S_\ell$ with~$\ell:=v_p - \lceil \log_2 k^2 \rceil$ is not contained in~$\cR(\cT_{p-1})$ and can
thus be safely added since $f(S_\ell) > \frac{k^2}{4}|\fmin|$ by Lemma~\ref{Lem:Int_new}.
For any $t=\ell-2,\ell-4,\dots,u_p(+1)$ set
\[
\cT_{p-1}^t ~:=~ \cT_{p-1} \cup \{S_t,S_{t+2},\dots,S_{\ell-2},S_\ell\}
\quad\text{and}\quad
\alpha_t ~:= \max_{S \in \cR(\cT_{p-1}^t)}f(S).
\]
We show by induction that $\alpha_t \leq 4f(S_t)$ for all $t=\ell,\ell-2,\ell-4,\dots,u_p(+1)$.
For $t = \ell$, we get with Lemma~\ref{Lem:max_ring_family} for the first inequality, and with $f(S_{\ell}) > \frac{k^2}{4}|\fmin|$ and $f(S_{\ell})>2|\fmin|$ by Lemma~\ref{Lem:Int_new} for the second inequality
\begin{align*}
\alpha_{\ell} 
~\leq~ 2\tfrac{k^2}{4}|\fmin| + 2|\fmin| + f(S_{\ell})
~<~ 2f(S_{\ell}) + f(S_{\ell}) + f(S_{\ell}) 
~=~ 4f(S_{\ell}).
\end{align*}
Assume that the claim is true for $t+2 \leq \ell$, i.e,
$\alpha_{t+2} \leq 4f(S_{t+2})$.
Then, again using Lemma~\ref{Lem:max_ring_family} for the first inequality and Lemma~\ref{Lem:Int_new} for the second inequality, we get
\begin{align*}
\alpha_t 
~\leq~ 2\bigl(4f(S_{t+2})+|\fmin|\bigr) + f(S_{t}) 
~<~ 2f(S_{t}) + 2|\fmin| + f(S_{t})
~<~ 2f(S_t) + f(S_t) + f(S_t) 
~=~4f(S_t).
\end{align*}
This concludes the induction. 
Since $f(S_{t}) > 4f(S_{t+2})$ by Lemma~\ref{Lem:Int_new}, $\alpha_{t+2} \leq 4f(S_{t+2})$ implies that $S_{t}$ is not contained in $\cR(\cT_{p-1}^t)$ and can thus be added to our sequence.

%
%
Summarizing, by Lemmas~\ref{Lem:Cardinality_I_1} and~\ref{Lem:Cardinality_I_2} we have constructed a chain of ring families of length 
\begin{align*}
\tfrac{1}{2}|I^3|-O(\log k)\cdot q 
~=~\tfrac{1}{2}|I^3|-O(\log k)\cdot(|I^1|+|I^2|)
~=~\tfrac{1}{2}|I^3|-O(k^2 \log k+m\log m\log k).
\end{align*}
By Theorem~\ref{Thm:Chain_of_ring_families}, the length of the chain is bounded
by~$O(k^2)$ and thus $|I^3|~\in~O(k^2 \log k+m\log m\log k) $.
\end{proof}

\subsection{Proofs of Theorem~\ref{Thm:Running_Time_Main} and Theorem~\ref{Thm:main}}
\label{subsec:theorem}

Now that we have derived bounds on the cardinalities of $I^1$, $I^2$, and $I^3$, we can prove Theorem~\ref{Thm:Running_Time_Main}.

\begin{proof}[Proof of Theorem~\ref{Thm:Running_Time_Main}]
The correctness of Algorithm~\ref{alg_main} has already been discussed at the very beginning of this
section.
It thus remains to prove the bound on the running time.
We start by analyzing the running time of one iteration of the while-loop.
Line~3 requires the minimization of the submodular function~$d^{\theta_i}$.
Using the submodular function minimization algorithm of Lee, Sidford, and Wong~\cite{Lee2015} with Orlin's min-cost flow algorithm to evaluate~$d^{\theta_i}$, this can be done in~$O\bigl(k^3 \log k \MCF\bigr)$ or $O\bigl(k^3 \log k \cdot m \log n (m+n\log\log n)\bigr)$ time.

Line 4 can be implemented in two steps: 
First, projecting the left-hand slope $\mincut^{\theta_i}(S_i)$ of $d^{\theta}(S_i)$ at time $\theta_i$
forward, we check whether $d^{\theta}(S_i)=0$ for~$\theta=\theta_i-
d^{\theta_i}(S_i)/\mincut^{\theta_i}(S_i)$.
If this happens to be the case, we have found $\theta_i'=\theta$.
The value $\mincut^{\theta_i}(S_i)$ can be determined in $O(\MCF)$ time by computing a static
minimum-cost circulation, and evaluating~$d$ takes $O\bigl(k^3 \log k \MCF\bigr)$ time, which adds up to
$O\bigl(k^3 \log k \MCF\bigr)$ in total.

Otherwise, we determine $\theta_i'$ by computing a quickest $s$-$t$-flow in the network $\dynN^S$.
Using the algorithm of Saho and Shigeno~\cite{Saho2017}, this takes $O(nm^2\log^2n)$ time.

Line 5 can be implemented with a binary search on~$J$ in~$O(\log\log k)$ iterations.
Each of these iterations is dominated by one submodular function minimization of~$d^\theta$ for
some~$\theta$.
Thus, Line 5 can be implemented to run in~$O\bigl(\log\log k(k^3 \log k \MCF)\bigr)$ time.

In the worst case, the algorithm of Saho and Shigeno needs to be used in Line 4 in each iteration~$i \in I^1 \cup I^2$.
Since $|I^1|+|I^2| \in O(k^2 + m\log m)$ by Lemmas~\ref{Lem:Cardinality_I_1} and~\ref{Lem:Cardinality_I_2}, the iterations in $I^1 \cup I^2$ require an overall worst case running time in $O\bigl((k^2 + m\log m)(nm^2\log^2n+\log\log k(k^3 \log k \cdot m \log n (m+n\log\log n))\bigr)$.

During an iteration~$i \in I^3$, the algorithm of Saho and Shigeno does not need to be used in Line 4 of the algorithm.
Thus, the running time of each iteration in $I^3$ is dominated by the running time of Line 5.
With $|I^3| \in O(k^2 \log k + m\log m \log k )$ by Lemma~\ref{Lem:Cardinality_I_3}, we achieve an overall running time for all iterations in~$I^3$ of $O\bigl((k^2 \log k + m\log m \log k)(\log\log k(k^3 \log k \cdot m \log n (m+n\log\log n)))\bigr)$ .

Overall, this adds up to a total running time of Algorithm~\ref{alg_main} in $\tilde{O}(k^5m^2+k^3m^3+m^3n)$.
\end{proof}
%
%
%
Our main result immediately follows.

\begin{proof}[Proof of Theorem~\ref{Thm:main}]
As proven above, the minimal feasible time horizon~$\theta^*$ can be determined with
Algorithm~\ref{alg_main} in $\tilde{O}(k^5m^2 + k^3m^3+m^3n)$ time.
Given~$\theta^*$, a feasible transshipment over time can be computed using a single call to 
Orlin’s submodular function minimization algorithm~\cite{Orlin2009} (or any other submodular function minimization algorithm relying on Cunningham’s framework) in time~$\tilde O(k^5 m^2)$)~\cite{Schloeter2017,schloter2018flows}.
Overall, this adds up to a running time in $ \tilde O(m^2k^5+m^3k^3+m^3n)$ to solve the Quickest Transshipment Problem.
\end{proof}

%% file: Figures/alg_main.tex
\begin{tikzpicture}[xscale = 0.8, yscale=0.65]
    \draw[->, >=stealth', thick] (-0.5,0)  -- (10,0);
    \draw[->, >=stealth', thick] (0, -3) -- (0,2) node(yline)[above] {$\scriptstyle d^\theta(S)$};

    \draw[thin, color = gray] (0,-1) -- (1,-1) -- (5,0.5)--(7,2);
    \draw[thin, color = gray] (0,-2.2)--(2,-2.2)--(5,-1)--(8,2);

    \draw[thick, color = black] (0,-2.5) -- (1.5,-2.5)--(3,-2)--(4,-1)--(5,0)--(5.7,1.7) node[above, black] {$\scriptstyle d^\theta(S_i)$} ;
    \draw[thin, color=gray] (0,-2)--(2,-2)--(5,-1.5)--(7,-0.5)--(9,1.2)--(9.3,1.7);
    \draw[thin, color=gray] (0,-1.5)--(1,-1.5)--(7,-1)--(9,1/3)--(10,1.5);
    
    \draw[thin, color=gray] (0,1)--(2.2,1.3)--(3.3,1.6)--(3.5,2);

    \draw[thin, color=gray] (0,0)--++(1,0.2)--++(0.5,0.5)--(1.7,2);

    \draw[thin, color=gray] (0,-1.25)--++(4.5,1)--++(3,1)--++(0.5,0.5)--(8.5,2);

    \draw[thick, color=black] (0,-0.2)--++(3,0.3)--++(1,0.8)-- (4.4,1.7) node[black, above] {$\scriptstyle d^\theta(S_{i-1})$};

    \draw[densely dashed, gray, name path = it1] (2,2) -- (2,-2.6) node[below, black]{$\scriptstyle\theta_{i}$};

    \draw[densely dashed, gray, name path = it2] (5,2) -- (5,-2.6) node[below, black]{$\scriptstyle\theta_{i}'$};

    \draw[densely dashed, gray, name path = it2] (8,2) -- (8,-2.6)node[below, black, xshift=-0.1cm] {$\scriptstyle\theta_{i+1}$};
    \draw[densely dashed, gray, name path = it2] (8.5,2) -- (8.5,-2.6)node[below, black, xshift = 0.1cm]{$\scriptstyle\theta^*$};

    \coordinate (m1) at (2,-2.33333){};
    \coordinate (m2) at (5,-1.5){};
    \coordinate (m3) at (7.5,-0.666666){};

    \coordinate (m4) at (8.5,0){};

    \fill[red] (m1) circle (1.5pt);
    \fill[red] (m2) circle (1.5pt);
    \fill[red] (m4) circle (1.5pt);

    \coordinate (p1) at (3.2 ,-1.8);
    \coordinate (p2) at (5.8 ,-1.1);
    \draw[thick, color = red] (0,-2.5) -- (1.5,-2.5)--(3,-2)--(p1)--(5,-1.5)--(p2)--(7,-1)--(8.5,0)--(9.8,0);


    \fill[red] (8,-0.33333) circle (1.5pt);

    \begin{scope}
          \clip(m2) rectangle (9.5,0.6);
          \draw[thick, color = green!50!black, name path = j1] (m2) -- ++(2,2) ;
          \draw[thick, color = green!50!black, name path = j2] (m2) -- ++(5,2.5);
          \draw[thick, color = green!50!black, name path = j3] (m2) -- ++(8,2);
          \draw[thick, color = green!50!black, dotted] (m2) -- ++(8,1);

        \fill[green!50!black] (8,0) circle (1.5pt); 
    \end{scope}



    \node[right] at (10,0) {\scriptsize time $\theta$};

\end{tikzpicture}

%% file: Figures/case_j_is_zero.tex
\begin{tikzpicture}[scale=1.3]
    \draw[->, >=stealth', thick] (1.5,0)  -- (6,0) node(xline)[right]{\scriptsize
        time $\theta$};
    \draw[->, >=stealth', thick] (2, -1) -- (2,1) node(yline)[above] {$\scriptstyle d^\theta(X)$};
    \draw[color = black, thick] (2,-0.8) -- (3,-0.7)--(3.5,-0.4)--(4,0)--++(0.25,0.2)--(4.58,1)node[black, above] {$\scriptstyle d^\theta(S_{i})$};;

  \draw[densely dashed, gray] (2.5,0.5) -- (2.5,-1);
  \draw[densely dashed, gray] (4,0.5) -- (4,-1);
  \draw[densely dashed, gray] (4.7,0.5) -- (4.7,-1);

    \draw[thick, green!50!black] (4,-0.7) --(6,0.9);

    \draw[densely dashed, thick](2.5,-0.75)--(4,0);

  \fill[red] (4,-0.7) circle (1.5pt){};
  \fill[red] (2.5,-0.75)  circle (1.5pt){};

    \node[below] at (4.7,-1) {$\scriptstyle\theta^*$};
      \node[below] at (2.5,-1) {$\scriptstyle\theta_i$};
      \node[below] at (4,-1) {$\scriptstyle\theta_{i+1}$};

    \draw[thick, color = red] (2,-0.8) -- (3,-0.7)--(4,-0.7)--(4.7,0)--(5.7,1);
\end{tikzpicture}